\documentclass[10pt]{article}
\usepackage{amssymb,amsmath,cite,bezier,hyperref}
\usepackage{algorithmic,algorithm}
\newtheorem{precor}{{\bf Corollary}}

\newenvironment{cor}{\begin{precor}{ \hspace{-0.5
               em}{\bf.\ }}}{\end{precor}}
\newtheorem{precon}{{\bf Conjecture}}

\newtheorem{predefin}{{\bf Definition}}

\newenvironment{defin}[1]{\begin{predefin}\rm{ \hspace{-0.5
                   em}{\bf.\ }}{\rm
#1}\hfill{$\spadesuit$}}{\end{predefin}}
\newtheorem{preexm}{{\bf Example}}

\newenvironment{exm}[1]{\begin{preexm}\rm{ \hspace{-0.5
                  em}{\bf.\ }}{\rm #1}\hfill{$\clubsuit$}}{\end{preexm}}

\newtheorem{preappl}{{\bf Application}}

\newtheorem{prelem}{{\bf Lemma}}

\newenvironment{lem}{\begin{prelem}{ \hspace{-0.5
               em}{\bf.\ }}}{\end{prelem}}
\newtheorem{preproof}{{\bf Proof.\ }}

\newenvironment{proof}[1]{\begin{preproof}{\rm
               #1}\hfill{$\blacksquare$}}{\end{preproof}}
\newtheorem{presproof}{{\bf Sketch of Proof.\ }}

\newenvironment{pproof}[2]{{\bf #1.} {\rm
               #2}\hfill{$\blacksquare$}}{\par}
\newtheorem{prethm}{{\bf Theorem}}

\newenvironment{thm}{\begin{prethm}{ \hspace{-0.5
               em}{\bf.\ }}}{\end{prethm}}
\newtheorem{prealphthm}{{\bf Theorem}}

\newenvironment{alphthm}{\begin{prealphthm}\rm{ \hspace{-0.5
               em}{\bf.\ }}}{\end{prealphthm}}
\newtheorem{preprop}{{\bf Proposition}}

\newenvironment{prop}{\begin{preprop}\rm{ \hspace{-0.5
               em}{\bf.\ }}}{\end{preprop}}
\newenvironment{prob}[3]{\vskip 5pt
\noindent {\bf #1}\\[2pt]
\hspace*{-6pt}\begin{tabular}{p{50pt}l}
{INSTANCE:}& \parbox[t]{11.8cm}{#2}\\[10pt]
 QUERY:    & \parbox[t]{11.8cm}{#3}}
{\end{tabular}\\[8pt]}
\newenvironment{probc}[4]{\vskip 5pt
\noindent {\bf #1}\\[2pt]
\hspace*{-6pt}\begin{tabular}{p{58pt}l}
CONSTANTS: & \parbox[t]{11.5cm}{#2}\\
{INSTANCE:}& \parbox[t]{11.5cm}{#3}\\
 QUERY:    & \parbox[t]{11.5cm}{#4}}{\end{tabular}\\[8pt]}
\def\isdef{\mbox {$\ \stackrel{\rm def}{=} \ $}}
\begin{document}
\def\thefootnote{\fnsymbol{footnote}}
\begin{center}
{\Large \bf  On Complexity of Isoperimetric Problems on Trees
}\\
\vspace*{0.5cm}
{\bf Amir Daneshgar\footnote{Correspondence should be addressed to {\tt
daneshgar@sharif.ac.ir}.}}\\
{\it Department of Mathematical Sciences} \\
{\it Sharif University of Technology} \\
{\it P.O. Box {\rm 11155--9415}, Tehran, Iran}\\
{\tt daneshgar@sharif.ir}\\ \ \\
{\bf Ramin Javadi}\\
{\it Department of Mathematical Sciences} \\
{\it Sharif University of Technology} \\
{\it P.O. Box {\rm 11155--9415}, Tehran, Iran}\\
{\tt rjavadi@mehr.sharif.ir}\\ \ \\
\end{center}
\begin{abstract}
This paper is aimed to investigate some computational aspects of different isoperimetric problems on weighted trees. 
In this regard, we consider different connectivity parameters called {\it minimum normalized cuts}/{\it isoperimteric numbers} defined through taking minimum of the maximum or the mean of the normalized outgoing  flows from a set of subdomains of vertices, where these subdomains constitute a {\it partition}/{\it subpartition}.
We show that the
decision problem for the case of taking $k$-partitions and the maximum
 (called the max normalized cut problem {\rm NCP}$^M$) as well as the other two decision problems for the mean version (referred to as {\rm IPP}$^m$ and {\rm NCP}$^m$) are $NP$-complete problems. On the other hand,  we show that the decision problem for the case of taking $k$-subpartitions and the maximum (called the max isoperimetric problem {\rm IPP}$^M$)
 can be solved in  {\it linear time} for any weighted tree  and any $k \geq 2$.
  Based on this fact, we provide polynomial time $O(k)$-approximation algorithms for
all different versions of $k$th isoperimetric numbers considered.\\
Moreover, when the number of partitions/subpartitions, $k$, is a fixed constant, as an extension of a result of B.~Mohar (1989) for the case $k=2$ (usually referred to as the Cheeger constant), we prove that
max and mean isoperimetric numbers of weighted trees as well as their max normalized cut can be computed in polynomial time.  We also prove some hardness results for the case of simple unweighted graphs and trees.
\\
\begin{itemize}
\item[]{{\footnotesize {\bf Key words:}\ isoperimetric number, Cheeger constant, normalized cut, graph partitioning, computational complexity, approximation algorithms, weighted trees.}}
\item[]{ {\footnotesize {\bf Subject classification:} 05C85, 68Q25, 68R10.}}
\end{itemize}
\end{abstract}
\def\thefootnote{\arabic{footnote}}
\setcounter{footnote}{0}
\section{Introduction}
The classical isoperimetric problem is a well-known and well-studied subject in Riemannian geometry, while the analogous problems in discrete case have recently been at the center of attention. Different aspects of these problems have been extensively studied in the literature and variety of relations to many important concepts have been discovered. The significance of the
isoperimetric problem is due to its relation to the central theoretical concepts and also its varied real world applications (e.g. see \cite{BAN80,CH97,LUB94,HAR04,LE91,MOH89,SM00,SOSA03} for motivations and the background).

Isoperimetric numbers can be considered as geometric tools to measure the connectivity of graphs. To begin, let us recall (e.g. see \cite{MOH89}) the definition of the
classical isoperimetric number (Cheeger constant) of a simple graph $G=(V,E)$ as
$$h(G) \isdef \displaystyle{\min_{_{|Q| \leq \frac{|V|}{2}}}}\ \frac{c(Q)}{|Q|}=\displaystyle{\min_{_{Q \subseteq V}}}\ \max \left \{\frac{c(Q)}{|Q|},\frac{c(Q)}{|Q^c|} \right \},$$
where
$$c(Q) \isdef | \{uv\in E: u \in Q\ \& \ v \not \in Q\}|,$$
and the not so common mean version as follows
\begin{equation}
\iota(G) \isdef \displaystyle{\min_{_{Q \subseteq V}}}\
\frac{c(Q)}{|Q||Q^c|}.
\end{equation}
Higher isoperimetric numbers, as  generalizations of the classical isoperimetric numbers, have been defined for a general Markov chain on a directed
base-graph and their properties has been studied extensively, (e.g. see \cite{DHJ10} and references therein).
These problems deal with minimizing the max/mean of the  normalized outgoing flows over all {\it subpartitions} (disjoint nonempty subsets) of the vertex set. One may also define similar parameters based on minimizing this value over all {\it partitions} of the vertex set usually known as the {\it minimum normalized cuts} (see e.g. \cite{SM00,SOSA03}).
Following the main result of \cite{DHJ10}, it is known that the isoperimetric numbers can be described as $\{0,1\}$-optimization programs which admit a relaxation to the reals, while this is not the case for the minimum normalized cuts. This fact can be considered as a clue that the normalized cut problem is likely to be harder than the isoperimetric problem, which is almost approved by the results of this article.

In this article we are going to investigate computational aspects of these parameters on weighted trees. Our motivations for this study are twofold. On the one hand, tree partitioning and in particular solving isoperimetric problems on weighted trees has its own importance due to the existence of many applications in the practical problems such as image segmentation and pattern recognition (e.g. see \cite{GR06,FTT04,LAMS09,BEPE95,CHSW07}). On the other hand, the study of isoperimetric problems on trees is important from a computational point of view, since they provide a universe in which by small perturbations of conditions, these problems change their computational hardness from simple (i.e. polynomial time) to hard (i.e. $NP$-complete)
and vise versa. In this regard, our results provide evidence to consider as a general belief that changing the problem from subpartitions to partitions or taking the mean instead of the maximum, usually makes the problem computationally harder.

Let us begin with a description of our general setup. Our framework is a {\it weighted graph} which is a simple graph $G=(V,E)$  along with two weight functions on the vertex and the edge sets as, $\omega:V\to {\mathbb Q}^+$ and $c:E\to {\mathbb Q}^+$, which is usually denoted by $G=(V,E,\omega,c)$. By an {\it unweighted graph} we mean a weighted graph where all the vertex and edge weights are equal to 1. For every nonvoid subsets $A,B\subseteq V$, we define
\[E(A,B)\isdef\{e=uv\in E:\ u\in A, v\in B\},\]
\[\omega(A)\isdef \sum_{u\in A} \omega(u),\quad c(A,B)\isdef \sum_{
e\in E(A,B)} c(e), \quad c(A)\isdef c(A,A^c). \]
The {\em normalized outgoing flow} of the set $A$ is defined as the quotient $c(A)/\omega(A)$.
The set ${\mathcal D}_{k}(V)$ is defined to be the set of all {\it $k$-subpartitions} $\{A_{{1}},\ldots,A_{{k}}\} \isdef \{A_{{i}}\}^k_1$ of $V$, where $A_{{i}}$'s are nonempty disjoint subsets of $V$.
The set of all {\it $k$-partitions} of a set $V$, which is denoted by ${\mathcal
P}_{{k}}(V)$, is the subclass of ${\mathcal D}_{{k}}(V)$ containing
all $k$-sets $\{A_{{i}}\}^k_1$ for which
$\cup_{{i=1}}^{{k}} A_{{i}}=V$. Also, for every positive integer $n$, the notation $[n]$ stands for the set $\{1,\ldots,n\}$.\\

Now, we define the mean and max isoperimetric numbers as well as the minimum normalized cuts as follows.
\begin{defin}{\label{DEFISO}
Given a weighted graph $G=(V,E,\omega,c)$, for each $k$, $1\leq k\leq |V|$,
the $k$th mean and max {\it isoperimetric numbers} of $G$, denoted by $\iota^{m}_{_{k}}(G)$ and $\iota^{M}_{_{k}}(G)$, respectively, are defined as
\begin{eqnarray*}
\iota^{m}_{_{k}}(G) &\isdef&
\displaystyle{\min_{{ \{A_{{i}}\}^{{k}}_{{1}} \in  {\mathcal
D}_{{k}}(V) }} } \ \frac{1}{k}  \left(
\displaystyle{\sum^{k}_{i=1}} \frac{
c(A_{{i}})}{\omega(A_{{i}})} \right),\\[2mm]
\iota^{M}_{_{k}}(G) &\isdef&
\displaystyle{\min_{{ \{A_{{i}}\}^{{k}}_{{1}} \in  {\mathcal
D}_{{k}}(V) }} } \ \max_{1\leq i\leq k}\
\frac{ c(A_{{i}})}{\omega(A_{{i}})}.
\end{eqnarray*}
Also, considering the partitions, we define the following related constants as the $k$th (mean and max) {\it minimum normalized cuts} of $G$,
\begin{eqnarray*}
\tilde{\iota}^{m}_{_{k}}(G) &\isdef& \displaystyle{\min_{{
\{A_{{i}}\}^{{k}}_{{1}} \in  {\mathcal P}_{{k}}(V) }} } \
\frac{1}{k} \left ( \displaystyle{\sum^{k}_{i=1}} \frac{
c(A_{{i}})}{\omega(A_{{i}})} \right ),\\[2mm]
\tilde{\iota}^{M}_{_{k}}(G) &\isdef& \displaystyle{\min_{{
\{A_{{i}}\}^{{k}}_{{1}} \in  {\mathcal P}_{{k}}(V) }} } \
 \max_{1\leq i\leq k}\
\frac{ c(A_{{i}})}{\omega(A_{{i}})}.
\end{eqnarray*}
We call a weighted graph $G$, mean (resp. max) {\it $k$-geometric}, if $\iota^m_k(G)=\tilde{\iota}^m_k(G)$ (resp. $\iota^M_k(G)=\tilde{\iota}^M_k(G)$). Also, $G$ is called mean (resp. max) {\it supergeometric}, if it is mean (resp. max) $k$-geometric for all $2\leq k\leq |V|$.
We call a vertex $v\in V$, a (mean or max) {\it $k$-outlier}, if there exists a minimizing subpartition achieving $\iota_k(G)$, where $v$ lies outside of the subpartition.
It is well-known that $\iota_2=\tilde{\iota}_2$  (see \cite{DHJ10}) and the common value is usually called the {\it Cheeger constant} or {\it edge expansion} in the literature.
}\end{defin}
In order to investigate computational complexity of these optimization problems, as is traditional for complexity results, we consider the corresponding decision problems. Also, since the isoperimetric parameters as operators on weight functions preserve scalar multiplication, without loss of generality, we assume that the range of all weight functions is $\mathbb Z$ (instead of $\mathbb Q$). Moreover, for simplicity we use a couple of notations. The acronyms IPP and NCP stand, respectively, for the isoperimetric and normalized cut problems. As before, the superscripts $m$ or $M$ determine the mean or max version of these problems, respectively\footnote{Note that whenever the superscripts $m$ and $M$ are omitted, it means that the statement is true for both versions.}, and subscript $k$ is used whenever $k$ is a constant and does not appear as part of the input.  For instance, NCP$^M_k$ refers to the following problem,
\begin{probc}
{NCP$\mathbf{^M_k}$}
{An integer $k$.}
{A weighted graph $G=(V,E,\omega,c)$ and a positive rational number $N\in {\mathbb Q}^+$.}
{Is it true that $\tilde{\iota}^M_k(G)\leq N$? In other words, is there a $k$-partition $\{A_{{i}}\}^k_1 \in {\cal P}_k(V)$ such that
$\displaystyle{\max_{1 \leq i \leq k}}\left\{ \frac{
c(A_{{i}})}{\omega(A_{{i}})}\right\} \leq N?$}
\end{probc}
By the following results,  the equivalent problems IPP$_{2}$ and NCP$_2$ are known to be $NP$-complete.
\begin{alphthm}\label{THMCMPLX}\\
{\rm (i) \cite{MOH89}} The problem {\rm NCP}$_2$ is $NP$-complete for (unweighted) general graphs with multiple edges.\\
{\rm (ii) \cite{SM00}} The problem {\rm NCP}$_2$ is $NP$-complete for bipartite planar weighted graphs.
\end{alphthm}
Note that, however, the planarity and the bipartiteness in Theorem~\ref{THMCMPLX}(ii) is not mentioned explicitly in \cite{SM00}, the above statement clearly follows  from the reduction provided in the proof. For a long time, it has been an open and challenging problem how well $\iota_2=\tilde{\iota}_2$ can be approximated in polynomial time for general graphs. The best current known result is due to Arora {\it et.~al.} which gives a polynomial time approximation algorithm that computes $\iota_2$ up to a factor of $O(\sqrt{\log n})$ for an $n$-vertex simple graph using semidefinite programming and geometric embedding (see \cite{ARV09,SHE09,VAZ01}). Also, Wu {\it et.~al.} present a polynomial time $((4+o(1))\log n)$-approximation algorithm for the minimum normalized cut on an $n$-vertex weighted graph \cite{WUCMS04}.

It is instructive to note that the {\it non-normalized} counterparts of the (mean) isoperimetric problem and the (mean) normalized cut problem are already known as
the {\it minimum $k$-subpartition} problem and {\it minimum $k$-way cut} problem, respectively (e.g. see \cite{NAIB08} for details and the background). Particularly, we know that there exists a  polynomial time $2(1-1/k)$-approximation algorithm for the minimum $k$-way cut problem which is based on
computation of the minimum $k$-subpartition problem \cite{NK07}. In Section~\ref{SECINEQ}, along the same lines, we prove a couple of basic inequalities (Theorem~\ref{THMINEQ}) which show that the isoperimetric numbers can be considered as an approximation for the minimum normalized cuts. In Section~\ref{SECACA} we consider the computational aspects of this approximation on weighted trees and we determine the computational complexity of the four main
isoperimetric and normalized cut problems. There we prove that IPP$^m$, NCP$^m$, and NCP$^M$ are all $NP$-complete for weighted trees, however, quite unexpectedly, it turns out that IPP$^M$ is a linear time solvable problem in this case. This is used to provide polynomial time $O(k)$-approximation algorithms for the $k$th isoperimetric number and the $k$th minimum normalized cuts.

In Section~\ref{SECFXDk} we focus on the case when the number of parts, $k$, is fixed and does not appear as part of the input.
For $k=2$, Mohar  \cite{MOH89} has proved that there exists a linear time algorithm that computes $\iota_2$ for trees.  In this section
as a generalization of Mohar's result we prove that, for each $k \geq 2$, all parameters $\iota^{M}_{{k}}$, $\iota^{m}_{{k}}$ and $\tilde{\iota}^{M}_{{k}}$ can be computed in polynomial time for weighted trees. We also show that this fact can not be extended to weighted graphs with bounded tree-width (unless $P=NP$!) by proving that for every fixed $k \geq 2$, {\rm IPP}$_k$ and {\rm NCP}$_k$ (in both max and mean versions) are $NP$-complete for bipartite weighted graphs with tree-width $2$.

In Section~\ref{SECUNIT}, we try to improve the hardness results to the case of unweighted (simple) graphs or trees. In this regard,
we provide a general reduction method that can be used to improve any known strong $NP$-completeness result for weighted graphs to an $NP$-completeness result for unweighted graphs.   Particularly, we use this reduction to prove the $NP$-completeness of {\rm NCP}$^M$ for unweighted trees and {\rm IPP}$_k$ and {\rm NCP}$_k$  for unweighted graphs.
\section{\label{SECINEQ}A basic inequality}
Our main result in this section is the following inequalities, which are counterparts of a similar result for the minimum k-way cut problem, that has already been proved in \cite{NK07}.
\begin{thm}\label{THMINEQ}
For every connected weighted graph $G$ and every integer $3\leq k\leq |V(G)|$,
\[\iota^M_k(G)\ \leq\ \tilde{\iota}^M_k(G)\ <\ (k-1)\ {\iota}^M_k(G),\]
\[\iota^m_k(G)\ \leq\ \tilde{\iota}^m_k(G)\ <\ 2(1-\frac{1}{k})\ \iota^m_k(G).\]
\end{thm}
Note that, when $k=2$, we have  $\iota_2(G)=\tilde{\iota}_2(G)$ for both max and mean versions \cite{DHJ10}.
Moreover, the result shows that the parameters $\iota_k^m(G)$ and $\iota_k^M(G)$ can be seen as approximations of the parameters $\tilde{\iota}^m_k(G)$ and $\tilde{\iota}^M_k(G)$, respectively. Therefore, from this point of view, the isoperimetric numbers can be considered as approximations for the minimum normalized cuts. We shall elaborate the computational aspects of these approximations in the forthcoming sections.
To prove  Theorem~\ref{THMINEQ}, we need the following lemma.
\begin{lem}\label{LEMINEQ}
Given an integer $k \geq 1$ and nonnegative numbers $\lambda, a_i, b_i \ (1\leq i\leq k)$, such that $0 < \lambda < k$ and $\sum_i a_i=1$, the following inequality holds,
\begin{equation}\label{INEQ}
\sum_{i=1}^k a_i b_i \leq \max_j{\left(\lambda a_j b_j+(1-\frac{\lambda}{k})b_j\right)}.
\end{equation}
Equality holds if and only if either for each $i$, $b_i=0$, or for each $i$ and some constant $b$,  $a_i={1}/{k}$ and $b_i=b$.
\end{lem}
\begin{proof}
{Let $I:=\sum_i a_ib_i$ and for every $1\leq j\leq k$, let $c_j:=\lambda a_jb_j+(1-\lambda/k)b_j$ and $t_j:=k^2a_j/\lambda+k/(k-\lambda)$. Then
\begin{eqnarray*}
\left(\sum_{j=1}^k t_j\right)\max_j (c_j-I)&\geq& \sum_{j=1}^k t_j
(c_j-I)\\
&=&\left(\frac{k(k-\lambda)}{\lambda}+\frac{k\lambda}{k-\lambda}-\frac{k^2} {\lambda}-\frac{k^2}{k-\lambda}\right)I+\sum_j\left(k^2a_j^2b_j+b_j\right)\\
&=&\sum_j\left(k^2a_j^2b_j+b_j-2k a_j b_j\right)=\sum_j (k a_j-1)^2b_j\geq 0.
\end{eqnarray*}
Thus, $\max_j (c_j-I)\geq 0$, as desired. Also, the equality conditions follow immediately from the proof.
}
\end{proof}
\begin{pproof}{Proof of Theorem~\ref{THMINEQ}}
{Lower bounds are trivial from the definitions. To prove the upper bounds, let  $\{A_i\}_1^k\in{\cal D}_k(V)$ be a $k$-subpartition of the vertices and define $A^*:=V\backslash (\cup_i A_i)$. For simplicity let $w_i:=\omega(A_i)$, $c_i:=c(A_i)$ and $C:=\sum_i c_i$.  For a fixed $j$\ $(1\leq j\leq k)$ define the $k$-partition $\pi^j:=\{B_i^j\}_{i}^k$ as $B_i^j:=A_i$ for all $i\neq j$ and $B_j^j:=A_j\cup A^*$.
Then, we have
\[c(B_j^j)\leq \sum_{i: i\neq j} c(A_i)=C- c_j.\]
Thus, for every $1\leq j\leq k$,
\begin{eqnarray}\label{INEQ1}
\max_i \left(\frac{c(B_i^j)}{\omega(B_i^j)}\right)&\leq& \max_{i: i\neq j}\left(\frac{c_i}{w_i}, \frac{C- c_j}{w_j+\omega(A^*)} \right),\\[2mm]
\label{INEQ2}
\sum_i \frac{c(B_i^j)}{\omega(B_i^j)}&\leq& \frac{C- c_j}{w_j+\omega(A^*)}+ \sum_{i:i\neq j}\frac{c_i}{w_i}.
\end{eqnarray}
Now, in order to prove the first inequality, assume that $G$ is not $k$-geometric (if $G$ is $k$-geometric the results are trivial) and let $\{A_i\}_1^k$ be a subpartition which achieves $\iota^M_k(G)$ and let $c_{j_0}=\max_i c_i$. By Inequality (\ref{INEQ1}), we have
\[\tilde{\iota}^M_k(G)\leq \frac{C- c_{j_0}}{w_{j_0}+\omega(A^*)}<
\frac{\sum_{i: i\neq j_0} c_i}{w_{j_0}}\leq (k-1)\frac{c_{j_0}}{w_{j_0}}\leq (k-1)\ \iota^M_k(G).\]
Now, in order to prove the second inequality, assume that $\{A_i\}_1^k$ be a subpartition which achieves $\iota^m_k(G)$. By Inequality (\ref{INEQ2}), we have
\begin{equation}\label{INEQ3}
k\ \tilde{\iota}^m_k(G)\leq \min_j\left(\frac{C- c_j}{w_j+\omega(A^*)}+\sum_{i:i\neq j}\frac{c_i}{w_i}\right)<
\min_j\left(\frac{C- 2 c_j}{w_j}\right)+\sum_i \frac{c_i}{w_i}.
\end{equation}
Now, let $C^*:=\sum_i (c_i/w_i)$, then applying Lemma~\ref{LEMINEQ}
with $a_j:=\frac{c_j/w_j}{C^*}$, $b_j:={w_j}$ and $\lambda:=2$, yields
\[\frac{C}{C^*}\leq \max_j \left(\frac {2c_j}{C^*}+(1-\frac{2}{k})w_j\right).\]
Therefore,
\begin{equation}\label{INEQ4}
\min_j\left(\frac{C-2c_j}{w_j}\right)\leq (1-\frac{2}{k}) \sum_i \frac{c_i}{w_i},
\end{equation}
and the result follows from Inequalities~(\ref{INEQ3}) and (\ref{INEQ4}).
}\end{pproof}
\begin{exm}{
In this example we show that the bounds in Theorem~\ref{THMINEQ} are sharp, in the sense that for every fixed $k\geq 3$, there is a family of weighted graphs $\{G_t\}_{t\in{\mathbb N}}$ such that $\tilde{\iota}^M_k(G_t)/\iota^M_k(G_t)$ tends to $(k-1)$ and $\tilde{\iota}^m_k(G_t)/\iota^m_k(G_t)$ tends to $2(1-\frac{1}{k})$ as $t$ tends to infinity.\\
Let $k$ be a constant. For every positive integer $t\geq k$, define the graph $G_t$ as a star with a central vertex $v$ of degree $k$ and $k$ vertices $v_1,\ldots, v_k$ each of degree 1. Also, define the weight functions $\omega$ and $c$ as follows,
\[\omega(v):=k,\ \omega(v_i):=t,\
c(vv_i):=1 \quad \forall\ 1\leq i\leq k.\]
Then, by the definitions we have
\[
\begin{array}{ll}
\displaystyle\iota^M_k(G_t)= \frac{1}{t},&\quad\displaystyle\tilde{\iota}^M_k(G_t)= \max\left(\frac{1}{t}, \frac{k-1}{t+k}\right)=\frac{k-1}{t+k},\\[10pt]
\displaystyle\iota^m_k(G_t)=\frac{1}{k}\sum_{i=1}^k \frac{1}{t}= \frac{1}{t},&\quad \displaystyle\tilde{\iota}^m_k(G_t)= \frac{1}{k}\ \left(\frac{k-1}{t+k}+\sum_{i=1}^{k-1}\frac{1}{t}\right)=
(1-\frac{1}{k})(\frac{1}{t+k}+\frac{1}{t}),
\end{array}
\]
where $\iota^M_k(G_t)$ and $\iota^m_k(G_t)$ are achieved for the disjoint sets $A_i:=\{v_i\}$, $1\leq i\leq k$, and $\tilde{\iota}^M_k(G_t)$ and $\tilde{\iota}^m_k(G_t)$ are achieved for the $k$-partition $\{B_i\}_1^k$, with $B_i:=\{v_i\}$, $1\leq i\leq k-1$ and $B_k:=\{v_k,v\}$.
The claim immediately follows from the above equalities.
}\end{exm}
\section{Algorithms, Complexity and Approximation Results}\label{SECACA}
In this section we consider IPP, NCP and their approximations  for weighted trees. In this regard, we shall prove that NCP$^M$ for weighted trees is $NP$-complete in the strong sense. Also, as a bit of a surprise, we show that the corresponding problem on subpartitions, i.e. IPP$^M$, can be solved in linear time using dynamic programming, where this can be used to obtain a polynomial time approximation for the minimum normalized cut of weighted trees.\\
 Let us recall that a problem with numerical parameters is said to be {\it $NP$-complete in the strong sense}, when it remains $NP$-complete, even when all of its numerical parameters are bounded by a polynomial in terms of the
length of the input. In other words, a strongly $NP$-complete problem remains $NP$-complete even when the input parameters are given in unary codes (instead of binary codes).
\begin{thm}\label{THMNCPM}
The problem {\rm NCP}$^M$  is $NP$-complete in the strong sense for weighted trees.
\end{thm}
\begin{proof}{
Clearly, {\rm NCP}$^M$  is in  $NP$. To prove the strong $NP$-completeness of the problem
we prove a sequence of reductions as follows,
$$\mbox{3-PARTITION}\ \leq^{^p}_{_{m}}\ \mbox{SUBSET AVERAGE}\ \leq^{^p}_{_{m}}\  {\rm NCP}^M,$$
where the well-known {3-PARTITION} problem and the SUBSET AVERAGE problem are defined as,
\begin{prob}{3-PARTITION}
{A positive integer $B\in {\mathbb Z}^+$ and $3m$ positive integers $x_1,\ldots,x_{3m}\in {\mathbb Z}^+$, such that $B/4< x_i< B/2$, for each $1\leq i\leq 3m$ and $\sum_{i=1}^{3m} x_i=mB$.}{Is there an $m$-partition $\{S_i\}^m_1 \in {\mathcal P}_m([3m])$ such that, for each $1\leq j\leq m$, $\sum_{i\in S_j} x_i=B$?}
\end{prob}
and
\begin{prob}{SUBSET AVERAGE}
{Positive integers $y_1,\ldots,y_n\in {\mathbb Z}^+$, where their average is an integer $\alpha$ along with a positive integer $m\leq n$.}
{Is there an $m$-partition $\{T_i\}^m_1  \in {\mathcal P}_m([n])$ such that, for each $1\leq j\leq m$, average of the elements with indices in $T_j$ is equal to $\alpha$, i.e. $\sum_{i\in T_j}y_i=\alpha{|T_j|}$?}
\end{prob}
\ \\
Note that the {3-PARTITION} problem is known to be strongly  $NP$-complete \cite{GJ79}, and consequently, the claim follows from the above reductions.\\[6pt]
{\bf Step 1.} $\mbox{3-PARTITION} \leq^{^p}_{_{m}} \mbox{SUBSET AVERAGE}$.\\
In the first step, we show that {SUBSET AVERAGE} is $NP$-complete in the strong sense, by a reduction from 3-PARTITION. Given $3m$ positive integers $x_1,\ldots,x_{3m}$ as an instance of 3-PARTITION, define for $1\leq i\leq 3m$,  $y_i:=x_i+B+1$ and for $3m+1\leq i\leq 4m$, $y_i:=1$. Now, consider $\{y_1,\ldots,y_{4m}\}$ together with the integer $m$ as an instance of {SUBSET AVERAGE}. The average of $y_i$'s is equal to $B+1$. If the answer to 3-PARTITION is yes, then there exists an $m$-partition $\{S_i\}^m_1 \in {\mathcal P}_m([3m])$ such that, for each $1\leq j\leq m$, $\sum_{i\in S_j} x_i=B$. Since $B/4< x_i< B/2$, each $S_j$ contains exactly $3$ elements. Now, by defining $T_j:=S_j\cup \{3m+j\}$, we have $\sum_{i\in T_j}y_i=4B+4=(B+1)|T_j|$. Hence, the answer to {SUBSET AVERAGE} is also yes.

On the other hand, assume that the answer to {SUBSET AVERAGE} is yes, then there exists an $m$-partition  $\{T'_i\}^m_1 \in {\mathcal P}_m([4m])$ such that, for each $1\leq j\leq m$, $\sum_{i\in T'_j} y_i=(B+1)|T'_j|$. Since $x_i$'s are positive, each $T'_j$ contains at least one of the elements $y_{3m+1},\ldots,y_{4m}$ and since there are $m$ disjoint subsets $T'_j$'s, each $T'_j$ contains exactly one of them. Thus, by defining $S'_j:= T'_j\backslash \{3m+1,\ldots,4m\}$, we have $\sum_{i\in S'_j}x_i=B$. Hence, the answer to {3-PARTITION} is also yes. This completes the reduction.\\[6pt]
{\bf Step 2.}$\ \mbox{SUBSET AVERAGE} \leq^{^p}_{_{m}}  {\rm NCP}^M$.\\
In the second step, we give a reduction from SUBSET AVERAGE to NCP$^M$ on weighted trees, where all the edge weights are equal to 1. Consider positive integers $y_1,\ldots,y_n$ with the average $\alpha$ and a positive integer $m \leq n$ as an instance of SUBSET AVERAGE. Also, let $l$ be an arbitrary positive fixed integer and construct a weighted tree $T=(V,E,\omega,c)$  as follows (see Figure~\ref{FIGA}).
\begin{gather*}
V:=\{u,u_i,v_{ij}\ |\  i=1,\ldots,n,\ j=1,\ldots,l-1\},\\
E:=\{uu_i, u_iv_{ij}\ |\ i=1,\ldots,n,\ j=1,\ldots,l-1\},\\
\omega(u):=n\alpha,\ \ \omega(u_i):=l y_i,\ \ \omega(v_{ij}):=\alpha,\ \forall\ 1\leq i\leq n,\ 1\leq j\leq l-1.
\end{gather*}
\begin{figure}[ht]
{\def\emline#1#2#3#4#5#6{%
\put(#1,#2){\special{em:moveto}}%
\put(#4,#5){\special{em:lineto}}}
\def\newpic#1{}
\unitlength .5pt
\special{em:linewidth 0.4pt}
\linethickness{0.4pt}
\begin{center}
\begin{picture}(200,220)(0,-120)
\bezier{500}(-19, -32)(53, 25)(126, 82)
\put(53,41){\makebox(0, 0)[cc]{}}
\bezier{500}(-19, -32)(53, -14)(126, 4)
\put(53,2){\makebox(0, 0)[cc]{}}
\bezier{500}(126, 4)(176, 2)(226, 0)
\put(176,18){\makebox(0, 0)[cc]{}}
\bezier{500}(-19, -32)(53, -70)(126, -109)
\put(53,-54){\makebox(0, 0)[cc]{}}
\bezier{500}(126, 82)(176, 99)(226, 116)
\put(176,115){\makebox(0, 0)[cc]{}}
\bezier{500}(126, 4)(176, 14)(226, 24)
\put(176,30){\makebox(0, 0)[cc]{}}
\bezier{500}(126, 4)(176, -16)(226, -36)
\put(176,0){\makebox(0, 0)[cc]{}}
\bezier{500}(126, -109)(176, -126)(226, -143)
\put(176,-110){\makebox(0, 0)[cc]{}}
\bezier{500}(126, 82)(176, 88)(226, 93)
\put(176,104){\makebox(0, 0)[cc]{}}
\bezier{500}(126, 82)(176, 69)(226, 55)
\put(176,85){\makebox(0, 0)[cc]{}}
\bezier{500}(126, -109)(176, -93)(226, -77)
\put(176,-77){\makebox(0, 0)[cc]{}}
\bezier{500}(126, -109)(176, -105)(226, -101)
\put(176,-89){\makebox(0, 0)[cc]{}}
\put(-19,-32){\circle*{16}}
\put(-19,-16){\makebox(0, 0)[cc]{$u$}}
\put(126,82){\circle*{16}}
\put(126,98){\makebox(0, 0)[cc]{$u_{_1}$}}
\put(126,4){\circle*{16}}
\put(126,20){\makebox(0, 0)[cc]{$u_{_2}$}}
\put(226,-11){\makebox(0, 0)[cc]{$\vdots$}}
\put(226,-36){\circle*{16}}
\put(263,-36){\makebox(0, 0)[cc]{$v_{_{2(l-1)}}$}}
\put(126,-40){\makebox(0, 0)[cc]{$\vdots$}}
\put(126,-109){\circle*{16}}
\put(126,-93){\makebox(0, 0)[cc]{$u_{_n}$}}
\put(226,116){\circle*{16}}
\put(256,116){\makebox(0, 0)[cc]{$v_{_{11}}$}}
\put(226,24){\circle*{16}}
\put(256,24){\makebox(0, 0)[cc]{$v_{_{21}}$}}
\put(226,0){\circle*{16}}
\put(256,0){\makebox(0, 0)[cc]{$v_{_{22}}$}}
\put(226,-115){\makebox(0, 0)[cc]{$\vdots$}}
\put(226,-143){\circle*{16}}
\put(263,-143){\makebox(0, 0)[cc]{$v_{_{n(l-1)}}$}}
\put(226,93){\circle*{16}}
\put(256,93){\makebox(0, 0)[cc]{$v_{_{12}}$}}
\put(226,80){\makebox(0, 0)[cc]{$\vdots$}}
\put(226,55){\circle*{16}}
\put(263,55){\makebox(0, 0)[cc]{$v_{_{1(l-1)}}$}}
\put(226,-101){\circle*{16}}
\put(256,-101){\makebox(0, 0)[cc]{$v_{_{n2}}$}}
\put(226,-77){\circle*{16}}
\put(256,-77){\makebox(0, 0)[cc]{$v_{_{n1}}$}}
\end{picture}
\end{center}}
\caption{A weighted tree corresponding to an instance of SUBSET AVERAGE.\label{FIGA}}
\end{figure}
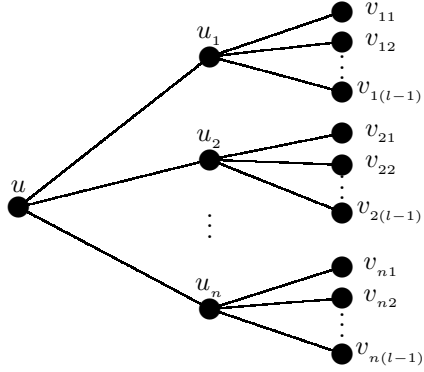

\noindent
Also, let all the edge weights be equal to $1$. The weighted tree $T$ together with the constants $k:=n(l-1)+m+1$ and $N:=1/\alpha$ constitute an instance of  NCP$^M$.   By assuming the partition $\{T_i\}^m_1 \in {\mathcal P}_m([n])$ as a positive answer to SUBSET AVERAGE, we define the $k$-partition
$$\{A_0\} \cup \{A_t\}^m_1 \cup \{A_{ij} \ \ | \ \ 1\leq i\leq n,  1\leq j\leq l-1 \} \in {\mathcal P}_k(V)$$
 as follows,
\[A_0:=\{u\},\ \ A_t:=\{u_i| i\in T_t\},\ \forall\ 1\leq t\leq m,\ \ A_{ij}:=\{v_{ij}\},\ \forall\ 1\leq i\leq n, 1\leq j\leq l-1.\]
Now, we have
\[\frac{c(A_0)}{\omega(A_0)}= \frac{n}{n\alpha}, \ \ \frac{c(A_t)}{\omega(A_t)}=\frac{l|T_t|}{\sum_{i\in T_t} l y_i}=\frac{1}{\alpha},\ \ \frac{c(A_{ij})}{\omega(A_{ij})}=\frac{1}{\alpha},\]
and consequently, the answer to NCP$^M$ is also yes.

On the other hand, assume that $\{A'_i\}^k_1$ be a positive answer to NCP$^M$. We should find a positive answer to SUBSET AVERAGE. In this regard, we come up with a partition of $[n]$ into at least $m$ subsets, each of which with an average equal to $\alpha$ (then, if it is necessary, we may merge some subsets and find an $m$-partition). Since $|V|=nl+1$, we have $|A'_j|\leq n-m+1$ and there are at least $m$ sets $A'_j$ which has non-empty intersection with the set $\{u_i\}^n_1$. Now, define $T'_j:=\{i|\ u_i\in A'_j\}$. Among $T'_j$'s, the non-empty ones form a partition of the set $[n]$. We claim that the average of each set in this partition is equal to $\alpha$. Fix $j$, where $T'_j$ is non-empty and let $\sigma(T'_j):=\sum_{i\in T'_j}y_i$. Since $|A'_j|\leq n-m+1$, we have
\begin{equation}
\frac{1}{\alpha}\geq \frac{c(A'_j)}{\omega(A'_j)}\geq \frac{l|T'_j|-(n-m)}{l \sigma(T'_j)+(2n-m-1)\alpha}. \label{EQ1c}
\end{equation}
Now, we choose $l$ sufficiently larger than $m,n,\alpha$, such that
\begin{equation}
 \frac{|T'_j|}{\sigma(T'_j)}- \frac{l|T'_j|-(n-m)}{l \sigma(T'_j)+(2n-m-1)\alpha}< \frac{1}{n\alpha^2}. \label{EQ2c}
\end{equation}
Note that $l$ depends only on $n,m,\alpha$ and does not depend on $j$ and $T'_j$, because $|T'_j|$ and $\sigma(T'_j)$ are respectively bounded by $n$ and $n\alpha$ (for instance one may choose $l=\alpha^3n^2(3n-2m-1)$).
Since $\sigma(T'_j)\leq n\alpha$, Equations (\ref{EQ1c}) and (\ref{EQ2c}) yield
\[\frac{|T'_j|}{\sigma(T'_j)}<\frac{1}{\alpha}+\frac{1}{n\alpha^2}\leq\frac{1}{\alpha}+\frac{1}{\alpha\ \sigma(T'_j)}.\]
Hence, $|T'_j|/\sigma(T'_j)\leq 1/\alpha$ and this shows that the average of integers $(y_i:\ i\in T'_j)$ is at least $\alpha$. Finally, since non-empty sets $T'_j$'s form a partition of $[n]$, the average of integers $(y_i:\ i\in T'_j)$ is exactly equal to $\alpha$.
This completes the reduction and hence NCP$^M$ is $NP$-complete in the strong sense.
}\end{proof}
Although Theorem~\ref{THMNCPM} can be considered as an evidence for hardness of NCP$^M$ for weighted trees, it turns out that the corresponding problem for subpartitions, i.e. IPP$^M$, is surprisingly a tractable problem. To prove this, we begin by the following lemma.
\begin{lem}\label{LEMCON}
Given a weighted graph $G=(V,E,\omega,c)$ and integer $k\geq2$,
there exists a minimizing subpartition $\{A_{i}\}_1^k\in {\cal D}_{k}(V)$ attaining ${\iota}_{k}(G)$ such that the induced graph on each $A_i$ is connected.
\end{lem}
\begin{proof}{
Let $\{A_i\}_1^k$ be a minimizing subpartition achieving ${\iota}_{k}(G)$ and assume that the induced graph $G$ on $A_{1}$ is not connected. Therefore $A_{1}=A\sqcup B$, where there is no edges between $A$ and $B$, we have
$$\min\left\{\frac{c(A)}{\omega(A)},\frac{c(B)}{\omega(B)}\right\}\leq \frac{c(A)+c(B)}{\omega(A)+\omega(B)}=\frac{c(A_{1})}{\omega(A_{1})}.$$
Hence, we may remove one of the sets $A$ or $B$ from $A_1$, such that the resulting subpartition remains minimizing. By continuing this process, we can find a minimizing subpartition with connected components.
}\end{proof}
\begin{thm}\label{THMIPPM}
There is a polynomial time  algorithm that decides {\rm IPP}$^M$ for weighted trees whose arithmetic complexity is $O(n)$.
\end{thm}
\begin{proof}{
We prove a stronger version of the theorem. We assume that in addition to the vertex and the edge weight functions, $\omega,c$, there exists another weight function $\gamma: V(T)\to \mathbb{Q}$ that intuitively can be considered as outgoing flows to the ground. Therefore, for every $A\subset V$, we define the outgoing flow from $A$ as $c(A):= \sum_{e\in E(A,A^c)} c(e)+\sum_{v\in A}\gamma(v)$ and we consider IPP$^M$ for these new weighted trees. It is clear that when $\gamma(v)=0$ for each $v\in V(T)$, the problem is the same as the classical IPP$^M$ introduced before.
Now, given a weighted tree $T=(V,E,\omega,c,\gamma)$ on $n$ vertices, an integer $k\geq 2$ and a number $N$ as the input of IPP$^M$, we perform the following algorithm on $T$ to decide if $\iota^M_k(T)\leq N$ and to find a proof (affirmative subpartition) if there exists any.

Let $v\in V$ be an arbitrary vertex and consider the rooted tree $T$ with the root $v$. Sort the vertices of $T$ as $v_1,\ldots,v_n=v$, in a way that the vertices at level $i+1$ precede the vertices at level $i$, for each $i$. This can be done in linear time by a breadth-first search. Note that in the following algorithm $\eta: V\to P(V)$ is a set function.

\begin{algorithm}[H]
\caption{Solve IPP$^M$}
\label{ALGIPPM}
{ \begin{algorithmic}
 \FOR{$i=1$ \TO $n$}
 \STATE $\eta(v_i)\leftarrow\{v_i\}$
 \ENDFOR
 \STATE $i,j\leftarrow 1$
 \WHILE{$j<k$ \AND $i\leq n$}
 \STATE Let $u$ be the unique parent of $v_i$ and $e:=uv_i\in E$ (if $i=n$, then $c(e)\leftarrow 0$)
 \IF{$\gamma(v_i)+c(e)\leq N \omega(v_i)$}
 \STATE $j\leftarrow j+1$, $A_j\leftarrow\eta(v_i)$, $\omega(A_j)\leftarrow\omega(v_i)$, $c(A_j)\leftarrow c(e)+\gamma(v_i)$, $\gamma(u)\leftarrow\gamma(u)+c(e)$
 \ELSIF{$\gamma(v)-c(e)<N \omega(v)$}
 \STATE $\eta(u)\leftarrow \eta(u)\cup \eta(v_i)$, $\omega(u)\leftarrow \omega(u)+\omega(v_i)$, $\gamma(u)\leftarrow\gamma(u)+\gamma(v_i)$
 \ELSE[i.e. $\gamma(v)-c(e)\geq N \omega(v)$]
 \STATE $\gamma(u)\leftarrow\gamma(u)+c(e)$
 \ENDIF
 \ENDWHILE
 \IF{j=k}
 \STATE \RETURN YES, $\iota_k^M(T)=\max_j\{c(A_j)/w(A_j)\}$ and $\{A_1,\ldots,A_k\}$
 \ELSE
 \STATE \RETURN NO
 \ENDIF
 \end{algorithmic}}
\end{algorithm}
Now, we prove the correctness of the algorithm. First, we fix a couple of notions. We say  two instances $(G_1,k_1,N_1)$ and $(G_2,k_2,N_2)$ are {\it equivalent} if the answer to IPP$^M$ for both of them are the same. Given a weighted graph $G=(V,E,\omega,c,\gamma)$ and a vertex $v\in V$, $G\backslash v=(V',E',\omega',c',\gamma')$ denotes the weighted graph obtained from $G$ by deleting the vertex $v$, where $\omega':=\omega|_{V'}$, $c':=c|_{E'}$ and for each $u\in V'$, $\gamma'(u):=\gamma(u)+\sum_{e=uv\in E} c(e)$. Also, for an edge $e\in E$, $G/e$ denotes the weighted graph obtained from $G$ by contracting the edge $e$, where the weight of the new vertex is defined as sum of the weights of the two old vertices. (If it is necessary we put together multiple edges and sum up their weights to get a simple graph.) Let $v$ be a leaf in $V(T)$ and $e=vu$ be the pendant edge.
\begin{enumerate}
\item If $\gamma(v)+c(e)\leq N \omega(v)$, then $(T,k,N)$ is clearly equivalent to $(T\backslash v,k-1,N)$.
\item If (1) is not the case and $\gamma(v)-c(e)<N \omega(v)$, then $(T,k,N)$ is equivalent to $(T/e,k,N)$. To see this, let $\pi:=\{A_i\}_1^k\in{\mathcal D}_k(V)$ be an affirmative answer for $T$, where the induced graph on each $A_i$ is connected (see Lemma~\ref{LEMCON}). If $u\not\in\cup A_i$, then $v\not\in\cup A_i$ (because $A_i$'s are connected) and hence, $\pi$ is also an affirmative answer for $T/e$. Now, assume that $u\in A_1$ and $v\not\in\cup A_i$. Define $A'_1:=A_1\cup\{v\}$, then,
\[c(A'_1)-N \omega(A'_1)=c(A_1)-N \omega(A_1)+\gamma(v)-c(e)-N \omega(v)<0.\]
Thus, the answer to $(T/e,k,N)$ is also yes.
\item Finally, if $\gamma(v)-c(e)\geq N \omega(v)$, then $(T,k,N)$ is equivalent to $(T\backslash v,k,N)$. To see this, as before let $\pi:=\{A_i\}_1^k\in{\mathcal D}_k(V)$ be an affirmative answer for $T$, where the induced graph on each $A_i$  is connected. If $v\not\in\cup A_i$, there is nothing to prove. If $v\in A_1$, then $u\in A_1$ (because $A_1$ is connected). Define $A'_1:=A_1\backslash v$, then,
\[c(A'_1)-N \omega(A'_1)=c(A_1)-N \omega(A_1)-\gamma(v)+c(e)+N \omega(v)\leq 0.\]
Thus, the answer to $(T\backslash v,k,N)$ is also yes.
\end{enumerate}
This shows that IPP$^M$ for weighted trees is self-reducible. Also, note that the runtime of the algorithm is clearly of order $O(n)$.
}\end{proof}
For an optimization problem, a {\it fully polynomial time approximation scheme} (FPTAS) is an algorithm that takes an instance of the problem together with a number $\epsilon>0$ and outputs a feasible solution within a factor $(1+\epsilon)$ of the optimal solution and its running time is bounded by a polynomial in the size of the instance and $1/\epsilon$.
By using Algorithm~\ref{ALGIPPM} as well as a standard iterative method, we can find an FPTAS to approximate $\iota^M_k(T)$. Also, using Theorem~\ref{THMINEQ}, we can find polynomial time approximation algorithms for the parameters $\tilde{\iota}^M_k(T)$, $\iota^m_k(T)$ and $\tilde{\iota}^m_k(T)$.
\begin{cor}Let $T$ be a weighted tree and $2\leq k\leq |V(T)|$ be an integer. \\
{\rm (i)} There exists an {\rm FPTAS} that approximates the parameter $\iota^M_k(T)$.\\
{\rm (ii)} For every $\epsilon>0$, there exists a polynomial time approximation algorithm that approximates the parameters $\tilde{\iota}^M_k(T)$,
$\iota^m_k(T)$ and $\tilde{\iota}^m_k(T)$, within factors $k-1+\epsilon$,  $k+\epsilon$ and  $2k-2+\epsilon$, respectively.
\end{cor}
\begin{proof}{
Given a weighted tree $T=(V,E,\omega,c)$, an integer $2\leq k\leq |V|$ and a number $\epsilon>0$, define $w_0:=\min_{v\in V}w(v)$, $W:=\sum_{v\in V}w(v)$, $c_0:=\min_{e\in E}c(e)$ and $C:=\sum_{e\in E}c(e)$. Therefore, $\iota_k^M(T)$ is within the interval $[2c_0/W,C/w_0]$. We start with this interval and do the following iteratively.\\
Let $[a_i,b_i]$ be the interval obtained in step~$i$. Then, in step~$i+1$, using Algorithm~\ref{ALGIPPM}, check if  $\iota_k^M(T)\leq (a_i+b_i)/2$ and find an interval containing  $\iota_k^M(T)$ whose length is $(b_i-a_i)/2$. We continue this process for $t$ steps, where $t:=\log(1/(2\epsilon))+\log(CW/c_0w_0-2)$. Finally, we come to an interval $[a_t,b_t]$ containing $\iota_k^M(T)$ whose length is $(C/w_0-2c_0/W)/2^t=\epsilon\ 2c_0/W$. We output $b_t$ as the approximation for $\iota_k^M(T)$. We have
\[\iota_k^M(T)\leq b_t= a_t+\epsilon\frac{2c_0}{W}\leq (1+\epsilon)\iota_k^M(T).\]
Also, the arithmetic complexity of this algorithm is $$O(nt)=O\left(n\left(\log(\frac{1}{2\epsilon})+\log(\frac{CW}{c_0w_0}-2)\right)\right),$$
and consequently, this is an FPTAS that approximates $\iota_k^M(T)$.\\
Part~(ii) follows from Part~(i), Theorem~\ref{THMINEQ} and the fact that $\iota_k^m(T)\leq\iota_k^M(T)\leq k\ \iota_k^m(T)$.
}\end{proof}
The next result (Theorem~\ref{THMIPPm}) shows that the
approximation method previously used to approximate
$\tilde{\iota}_k^M$ by $\iota_k^M$ can not be applied to approximate $\tilde{\iota}_k^m$ by $\iota_k^m$, since, contrary to the max version,
IPP$^m$ appears to be an $NP$-complete problem for weighted
trees. To prove this, first we need the following simple lemma
that will also be used in the proof of Theorem~\ref{THMCGG}.
\begin{lem}\label{LEMDIF}
Let $G=(V,E,\omega,c)$ be a connected weighted graph and
$S=\{v_1,\ldots,v_s\}\subset V$ be a fixed subset of vertices.
Define $W:=\sum_{u\in V\backslash S}\omega(u)$, $C:=\sum_{e\in
E}c(e)$ and $c_0:=\min_{e\in E}c(e)$. If $s\leq k\leq |V|$ is an
integer and for each $1\leq i\leq s$, $\omega(v_i)\geq
({2CW}/{c_0})$, then there exists a minimizing partition $($resp.
subpartition$)$ achieving $\tilde{\iota}_k(G)$ $($resp.
${\iota}_k(G))$ in which all the vertices $v_1,\ldots,v_s$ are in
different parts. Also, none of the vertices in $S$ are $k$-outlier.
\end{lem}
\begin{proof}
{We prove the lemma for ${\iota}^m_k(G)$. The other cases are
similar. Let $\{A_i\}_1^k$ be a minimizing subpartition achieving
${\iota}^m_k(G)$ and assume that $A_1$ contains two vertices in
$S$, say $v_1,v_2$ and $\omega(v_1)\geq \omega(v_2)\geq
({2CW}/{c_0})$. Then there is a subset, say $A_2$, which contains
no vertex of $S$. Now, move $v_2$ from $A_1$ to $A_2$ and call the
new subsets $A'_1$ and $A'_2$. Thus
\[\frac{c(A'_1)}{\omega(A'_1)}+\frac{c(A'_2)}{\omega(A'_2)}\leq \frac{2C}{\omega(v_2)}\leq
\frac{c(A_2)}{\omega(A_2)}<\frac{c(A_1)}{\omega(A_1)}+\frac{c(A_2)}{\omega(A_2)}.\]
This contradicts the fact that $\{A_i\}_1^k$ is a minimizing
subpartition. Therefore, all of the vertices $v_1,\ldots,v_s$ are
in different parts. Also, if a vertex $v_i$ does not lie in the
subpartition, we may add it to a subset $A_j$ which has no
intersection with $S$ to find a new subpartition contradicting the
minimality of $\{A_i\}_1^k$. Hence, no vertex in $S$ can be a
$k$-outlier. }\end{proof}
\begin{thm}\label{THMIPPm}
The problems {\rm IPP}$^m$ and {\rm NCP}$^m$ are $NP$-complete for weighted trees.
\end{thm}
\begin{proof}{
We verify a reduction from the $NP$-complete problem
EQUIPARTITION \cite{GJ79}.
\begin{prob}
{EQUIPARTITION}
{$2n$ positive integers $x_1,\ldots,x_{2n}$ such that
$\sum_{i=1}^{2n} x_i=2B$.}
{Is there a subset $I\subset [2n]$ such that
$|I|=n$ and $\sum_{i\in I} x_i=B$?}
\end{prob}
Consider positive integers $x_1,\ldots,x_{2n}$ with the sum $2B$
as an instance of EQUIPARTITION. Define $Q:=(1/2)\sum_{i=1}^{2n} x_i^2$ which is an integer
and construct a weighted tree $T=(V,E,\omega,c)$, where
$V:=\{v_0,v_1,\ldots, v_{2n},u_1,\ldots, u_{2n}\}$ and
$E:=\{v_0u_i, u_iv_i,\ i=1,\ldots,2n\}$. Also, let $k:=3n+1$ and for arbitrary
positive integers $d,D$, define the weight functions as follows.
\[
\omega(v_0):=2dB,\ \ \omega(v_i):=2D,\ \ \omega(u_i):=2x_i,\
\forall\ 1\leq i\leq 2n,\]
\[c_i:=c(v_0v_i)=c(u_iv_i)=x_i\left(
(d+1)^2B^2+Q-B x_i\right).
\]
Suppose that $d,D$ be sufficiently larger than $B$. Then by
Lemma~\ref{LEMDIF}, none of the vertices $v_0,v_1,\ldots,v_{2n}$
are $k$-outlier. Also, if for some $i$, the vertex $u_i$ is $k$-outlier,
then we can move $u_i$ to the set containing $v_i$, without
increasing the normalized outgoing flow of that set. Thus the tree
$T$ is $k$-geometric (i.e. $\iota_k(T)=\tilde{\iota}_k(T)$) and
in every minimizing $k$-partition, the vertices
$v_0,v_1,\ldots,v_{2n}$ lie in different parts. Moreover, suppose that
$D$ is sufficiently larger than $d$, then there exists a
minimizing $k$-partition in which each vertex $v_i$ forms a
single part in the partition. Thus, the minimizing partition which
achieves $\iota^m_k(T)=\tilde{\iota}^m_k(T)$ is of the form
\[\pi_I:=\{\ \{v_1\},\ldots,\{v_{2n}\}, \{v_0,u_i, i\in I\},
\{u_j\}, j\not\in I\ \},\] for some subset $I\subset [2n]$ with
$|I|=n$. Therefore, $k\ \iota^m_k(T)=k\ \tilde{\iota}^m_k(T)\leq N$ if and
only if there exists an $n$-subset $I\subset [2n]$, where
\begin{equation}\label{INEQB}
\sum_{i=1}^{2n} \frac{c_i}{\omega(v_i)}+\frac{\sum_{i=1}^{2n} c_i}{\omega(v_0)+\sum_{i\in I} \omega(u_i)}+ \sum_{i\not\in I} \frac{2c_i}{\omega(u_i)}\leq N.
\end{equation}
On the other hand,
\[
\sum_{i=1}^{2n} c_i=\left(
(d+1)^2B^2+Q\right)\sum_i x_i-B\sum_i x_i^2=2(d+1)^2B^3.
\]
Consequently, Inequality~(\ref{INEQB}) is equivalent to
\[\frac{(d+1)^2B^3}{D}+\frac{(d+1)^2B^3}{dB+\sum_{i\in I} x_i}+ \sum_{i\not\in I} \left(
(d+1)^2B^2+Q-B x_i\right)\leq N.
\]
If we define
\begin{equation}\label{INEQN}
N:= n(d+1)^2B^2+nQ+dB^2+\frac{(d+1)^2B^3}{D},
\end{equation}
then by substituting $N$ from (\ref{INEQN}) and simplifying, we have the following inequality.
\[(d+1)^2B^2\leq \left(dB+\sum_{i\in I} x_i\right)\left(dB+\sum_{i\not\in I} x_i\right).\]
Now, since $\sum_{i=1}^{2n} x_i=2B$, we have $(dB+\sum_{i\in I} x_i)(dB+\sum_{i\not\in I} x_i)\leq (d+1)^2B^2$ and equality holds if and only if there exists some $I$ such that $\sum_{i\in I}x_i=\sum_{i\not\in I}x_i=B$. Hence, $\iota^m_k(T)=\tilde{\iota}^m_k(T)\leq N/k$ if and only if there exists some subset $I$ with $|I|=n$, where $\sum_{i\in I} x_i=B$. This completes the proof.
}\end{proof}
\section[The Case of Fixed k]{The Case of Fixed $k$\label{SECFXDk}}
In this section we concentrate on the computation of the isoperimetric parameters when $k$ is assumed to be a constant. In fact the main theorem that we shall prove in this section is the following.
\begin{thm}\label{THMFXDk}
Let $k\geq 2$ be a constant integer. Then, there exists an algorithm that computes parameters $\iota^M_{k}(T)$ and $\iota_k^m(T)$ for every weighted tree $T=(V,E,\omega,c)$, in time $O(n^{\lfloor3k-1/2\rfloor}/(k-1)!)$.
Also, there exists an algorithm that computes  $\tilde{\iota}^M_k(T)$ in time $O((n+S(2k^2-6k-2,k))\ n^{(2k^2-6k-3)}/(k-2)!)$, where $S(N,K)$ is the Stirling number of second kind.
\end{thm}
Note that the running times of the algorithms presented in Theorem~\ref{THMFXDk} are exponential in $k$, but polynomial in $n$, when $k$ is a constant. Nevertheless, this exponential inefficiency is likely to be unavoidable duo to Theorems~\ref{THMNCPM} and \ref{THMIPPm}.\\
To prove this theorem, first we must prove a series of lemmas. To begin, we introduce the concept of the quotient of a graph $G=(V,E)$ with respect to a $k$-partition of $V$.
\begin{defin}{\label{RDC}
Given a weighted graph $G=(V,E,\omega,c)$ and a $k$-partition $\pi=\{A_i\}_{1}^{k}\in {\cal P}_{k}(V)$, for each $1\leq i\leq k$, let $\{A_{i}^{1},\ldots,A_{i}^{{n_i}}\}$ be the set of connected components of the induced graph of $G$ on $A_{i}$. The {\em quotient graph} of $G$ with respect to $\pi$, denoted by $G/\pi$, is defined to be a weighted graph $G/\pi=(V',E',\omega',c')$, where
\[V':=\{v_{i}^{{r}}:\ 1\leq i\leq k, 1\leq r\leq n_i \},\]
\[E':=\{v_{i}^{{r}}v_{j}^{{s}}:\ E(A_{i}^{{r}},A_{j}^{{s}})\neq \emptyset \},\]
\[\omega'(v_{i}^{r}):= \omega(A_{i}^{r}),\quad
c(v_{i}^{{r}}v_{j}^{{s}})':=
\sum_{e\in E(A_{i}^{{r}},A_{j}^{{r}})}c(e).\]
It is clear that the quotient graph $G/\pi$ is a minor of $G$ as a graph. Thus, if $G$ is planar, then $G/\pi$ is planar as well. Furthermore, if $G$ is acyclic, then $G/\pi$ is also acyclic. For a subset $F\subseteq E$, the graph obtained from $G$ by deleting the edges in $F$, is denoted by $G\backslash F$.
}\end{defin}
\begin{lem}\label{LEMQTNT}
Let $G=(V,E,\omega,c)$ be a weighted graph and $\pi=\{A_{i}\}_1^k\in {\cal D}_{k}(V)$ be a minimizing subpartition for $\iota_{k}(G)$. Define the $(k+1)$-partition ${\overline{\pi}}:=\{A_{i}\}_1^{k+1}$, where $A_{{k+1}}=V\backslash (\cup_1^k A_{i})$ and let $G/\overline{\pi}$ be the quotient graph of $G$ with respect to $\overline{\pi}$. Then, we have $\iota_{k}(G)=\iota_{k}(G/\overline{\pi})$. $($Similar statements are also true for the other parameters $\tilde{\iota}^m_{k}$ and $\tilde{\iota}^M_{k}$.$)$
\end{lem}
\begin{proof}{
We prove the lemma for $\iota_{k}^m$. The other cases follow similarly. Let $V',c',\omega'$ be as in Definition~\ref{RDC} and for every $1\leq i\leq k$,
define $A_{i}':=\{v_{i}^{{r}}:\ 1\leq r\leq n_i \}$. Then,
\[\iota^m_{k}(G)=\frac{1}{k}\sum_{i=1}^k\frac{c(A_{i})}{\omega(A_{i})}
=\frac{1}{k}\sum_{i=1}^k\frac{c'(A_{i}')}{\omega'(A_{i}')}\geq
\iota^m_{k}(G/\overline{\pi}).\]
Also, if $\pi'=\{B_{i}'\}_1^k\in{\cal D}_{k}(V')$ is a minimizing subpartition for $\iota^m_{k}(G/\overline{\pi})$ and $B_{i}
:=\cup\{A_{j}^{{s}}:\ v_{j}^{{s}}\in B_{i}'\}$, then,
\[\frac{1}{k}\sum_{i=1}^k\frac{c'(B_{i}')}{\omega'(B_{i}')}
=\frac{1}{k}\sum_{i=1}^k\frac{c(B_{i})}{\omega(B_{i})}\geq \iota^m_{k}(G).\]
}\end{proof}
\begin{lem}\label{LEMCOMP1}
Let $G=(V,E,\omega,c)$ be a weighted graph and $2\leq k\leq |V|$ be an integer. Then, there exists a subpartition $\pi=\{A_{i}\}_1^k\in {\cal D}_{k}(V)$ attaining ${\iota}_k(G)$ such that the number of connected components of $G\backslash \cup_{i}E(A_{i},A_{i}^c)$ is\\
{\rm (i)} at most $\lfloor {(3k-1)}/{2}\rfloor$, if $G$ is acyclic.\\
{\rm (ii)} at most $3k-4$, if $G$ is planar.
\end{lem}
\begin{proof}{
Consider the nonempty set $\mathcal{C}_k(V)$ of all the minimizing subpartitions $\{A_{i}\}_1^k\in {\cal D}_{k}(V)$ where the induced graph on each $A_{i}$ is connected (see Lemma~\ref{LEMCON}), and for each such subpartition, let $\{A_{{k+1}}^1,\ldots,A_{{k+1}}^{d}\}$ be the set of all connected components of the induced graph on $A_{{k+1}}:=V\backslash (\cup_1^k A_{i})$. Now, choose an extremal subpartition $\pi=\{A_{i}\}_1^k\in {\cal C}_{k}(V)$ for which $d$ is minimized.
Let $\overline{\pi}:=\{A_i\}_1^{k+1}$ and $V(G/\overline{\pi})=\{v_1,\ldots,v_k,v_{k+1}^{1},\ldots,v_{k+1}^{d}\}$ as in Definition~\ref{RDC}. First, we claim that $\deg(v_{{k+1}}^{p})\geq 3$ for each $1\leq p\leq d$. By contradiction, assume that $\deg(v_{{k+1}}^{{p}})\leq 2$. Then, $A_{{k+1}}^{{p}}$ is connected to at most two subsets in $\pi$, say $A_{1}, A_{2}$. Without loss of generality, assume that $c(A_{{k+1}}^{{p}},A_{1})\geq c(A_{{k+1}}^{{p}},A_{2})$. Define $B_{1}:= A_{1}\cup A_{{k+1}}^{{p}}$ and $B_{i}:= A_{i}$ for all $2\leq i\leq k$. Therefore, $\pi'=\{B_i\}_1^k\in {\cal D}_k(G)$ is a subpartition and
\[
\frac{c(B_{1})}{\omega(B_{1})}=\frac{c(A_{1})-c(A_{1},A_{{k+1}}^{{p}})+
c(A_{{k+1}}^{{p}},A_{2})}{\omega(A_{1})+\omega(A_{k+1}^{{p}})}< \frac{c(A_{1})}{\omega(A_{1})},\]
that contradicts the minimality of $\pi$.
 Hence, $\deg(v_{_{k+1}}^{p})\geq 3$, for each $1\leq p\leq d$ and the set of vertices of $G/\overline{\pi}$ with degree less than $3$ is a subset of $\{v_{_1},\ldots, v_{_k}\}$.

Let $G'$ be the graph obtained from $G/\overline{\pi}$ by deleting all the edges $e=v_iv_j\in E(G/\overline{\pi})$, for every $1\leq i,j\leq k$. Then,
\begin{equation}\label{EQDEG}
|E(G')|= \sum_{p=1}^d \deg(v_{k+1}^{p})\geq 3(|V(G')|-k).
\end{equation}
On the other hand, if $G$ is acyclic, then $G'$ is also acyclic and  $|E(G')|\leq |V(G')|-1$. This fact along with (\ref{EQDEG}) yields $|V(G/\overline{\pi})|=|V(G')|\leq {(3k-1)}/{2}$.\\
Now, if $G$ is planar, then $G'$ is also planar. Furthermore,
 $G'$ is bipartite with independent parts $\{v_1,\ldots,v_k\}$ and $\{v_{k+1}^{1},\ldots,v_{k+1}^{d}\}$.  Therefore, $G'$ is a bipartite planar graph and $|E(G')|\leq 2|V(G')|-4$. This fact along with Inequality~(\ref{EQDEG}) yields $|V(G/\overline{\pi})|=|V(G')|\leq 3k-4$.
}\end{proof}
\begin{lem}\label{LEMCOMP2}
Let $T=(V,E,\omega,c)$ be a weighted tree and $3\leq k\leq |V|$ be an integer.
Then, there exists a minimizing partition $\pi=\{A_{i}\}_1^k\in {\cal P}_{k}(V)$ for $\tilde{\iota}^M_k(T)$ such that the number of connected components of $T\backslash \cup_{i}E(A_{i},A_{i}^c)$ is at most $\max\{2k^2-6k-2,k\}$.
\end{lem}
\begin{proof}{
let $\pi=\{A_{i}\}_1^k\in {\cal P}_{k}(V)$ be a $k$-partition achieving $\tilde{\iota}^M_{k}(T)$
such that the number of vertices of $T/\pi$ is minimal.
Let $\{A_{i}^{1},\ldots,A_{i}^{{n_i}}\}$ be the set of connected components of the induced graph $T$ on $A_{i}$ and $V'$ be the set of vertices of $T/\pi$ as in Definition~\ref{RDC}. For each $i$, partition the set $[n_i]$ into two subsets $L_i$ and $L_i^c$, where $L_i:=\left\{r\ |\ 2\leq r\leq n_i,\ \frac{c(A_{i}^{{r}})}{\omega(A_{i}^{{r}})}\geq \frac{c(A_{i})}{\omega(A_{i})}\right\}$.\\
Firstly, we prove that for each $r\in L_i$, $\deg(v_{i}^{{r}})\geq 3$.
By contradiction, assume that $r\in L_i$ and $\deg(v_{i}^{{r}})\leq 2$. Therefore, $A_{i}^{{r}}$ is connected to at most two sets, say $A_{j}^{{s}}, A_{l}^{{t}}$. Without loss of generality assume that $c(A_{i}^{{r}},A_{j}^{{s}})\geq c(A_{i}^{{r}},A_{l}^{{t}})$. Now, let $B_{i}:= A_{i}\backslash A_{i}^{{r}}$, $B_{j}:= A_{j}\cup A_{i}^{{r}}$ and $B_{h}:= A_{h}$ for $h\neq i,j$. Thus, $\{B_{i}\}_1^k$ is a $k$-partition of $V$ and since $r\in L_i$,
\begin{eqnarray*}
\frac{c(B_{i})}{\omega(B_{i})}&=&\frac{c(A_{i})-c(A_{i}^{{r}})}
{\omega(A_{i})-\omega(A_{i}^{{r}})}\leq \frac{c(A_{i})}{\omega(A_{i})},\\ 
\frac{c(B_{j})}{\omega(B_{j})}&=&\frac{c(A_{j})-c(A_{i}^{{r}},A_{j}^{{s}})+
c(A_{i}^{{r}},A_{l}^{{t}})}{\omega(A_{j})+\omega(A_{i}^{{r}})}< \frac{c(A_{j})}{\omega(A_{j})}.
\end{eqnarray*}
Hence, $\max_i \left(\frac{c(B_{i})}{\omega(B_{i})}\right)\leq\tilde{\iota}^M_k(T)$ that contradicts the minimality of $|V(T/\pi)|$. Therefore, $\deg(v_{i}^{{r}})\geq 3$, whenever $r\in L_i$. Moreover, the above argument shows that if $k=3$, then for each $i$, $n_i=1$ and $|V(T/\pi)|=k=3$.\\
Secondly, provided $k\geq 4$, we prove that for each $i$, $|L_i^c|\leq k-3$. By contradiction, assume that $|L_1^c|\geq k-2$ and define $A':=A_2$ and $A'':=\cup_{3\leq i\leq k}A_i$. For each $r\in L_1$, as before, we transfer the vertices in $A_1^{r}$ into $A'$ or $A''$, without increasing the normalized outgoing flow of these subsets. Call the new subsets as $B'$ and $B''$. Hence, $\pi':=\{B',B'',A_1^r\ |\ r\in L_1^c\}$ is a $k$-partition that achieves $\tilde{\iota}_k^M(T)$ whereas $|V(T/\pi')|<|V(T/\pi)|$. This contradicts the minimality of $|V(T/\pi)|$ and therefore, $|L_i^c|\leq k-3$, for each $i$, whenever $k\geq 4$.\\
These facts show that the number of vertices in $V'=V(T/\pi)$ whose degrees are less than $3$, is at most $k(k-3)$. Hence,
\[2(|V'|-1)=2|E(T/\pi)|=\sum_{v\in V'} \deg(v)\geq 3(|V'|-k(k-3))+k(k-3),\]
and consequently, $|V(T/\pi)|\leq 2k(k-3)-2$, whenever $k\geq 4$.
}\end{proof}
\begin{pproof}{Proof of Theorem~\ref{THMFXDk}}
{Now, we are ready to provide algorithms to compute $\iota^m_{k}(T)$, $\iota^M_{k}(T)$ and $\tilde{\iota}^M_k(T)$ for a weighted tree $T$ and to find the corresponding minimizing partitions and subpartitions. We use Lemmas~\ref{LEMCOMP1},\ref{LEMCOMP2} and the fact that if we remove $t$ edges from a tree, we obtain a forest with exactly $t+1$ connected components. Our algorithm is as follows.

\begin{algorithm}[H]
\caption{Compute $\iota_k^m(T)$, $\iota_k^M(T)$ and $\tilde{\iota}_k^M(T)$}
Part 1. Compute $\iota_k^m(T)$ and $\iota_k^M(T)$.
\begin{algorithmic}[1]
\FOR{every subset $F\subset E(T)$ of size $k-1\leq |F|\leq \lfloor\frac{3k-3}{2}\rfloor$}
\STATE\label{LINECC} Find the connected components of $T\backslash F$, say $\{T_1,\ldots, T_t\}$, where $t=|F|+1$, compute the quantity $f_i:={c(A_i)}/{\omega(A_i)}$ for each $1\leq i\leq t$, where $A_i:= V(T_i)$.
 \STATE\label{LINESORT} Sort $f_1,\ldots,f_t$ in increasing order, say $f_{n_1}\leq f_{n_2}\leq\ldots\leq f_{n_t}$.
 \STATE\label{LINEIM} Define $\pi:=\{A_{n_1},\ldots, A_{n_k}\}$, $I^M(\pi):=f_{n_k}$ and $I^m(\pi):=(f_{n_1}+\ldots+f_{n_k})/k$.
 \ENDFOR
\STATE \RETURN The subpartition $\pi^M$ (resp. $\pi^m$) for which $I^M$ (resp. $I^m$) is minimum among all the subpartitions $\pi$ obtained in line \ref{LINEIM} \AND $I^M(\pi^M)$ (resp. $I^m(\pi^m)$).
\end{algorithmic}
Part 2. Compute $\tilde{\iota}_k^M(T)$.
\begin{algorithmic}[1]
\FOR{every subset $F\subset E(T)$ of size $k-1\leq |F|\leq \max\{2k^2-6k-3,k-1\}$}
\STATE Find the connected components of $T\backslash F$, say $\{T_1,\ldots, T_t\}$, where $t=|F|+1$, compute the quantity $f_i:={c(A_i)}/{\omega(A_i)}$ for each $1\leq i\leq t$, where $A_i:= V(T_i)$.
\STATE\label{LINEPI} Consider all possible ways that $\{A_i\}_1^t$ can be partitioned into $k$ subsets to form a $k$-partition of $V(T)$. Let $\pi:=\{B_{1},\ldots, B_{k}\}$ be such a $k$-partition and define $I^M(\pi):=\max_i\{c(B_i)/\omega(B_i)\}$.
\ENDFOR
\STATE \RETURN The $k$-partition $\pi^*$ for which $I^M$ is minimum among all the $k$-partitions $\pi$ obtained in line \ref{LINEPI} \AND $I^M(\pi^*)$.
\end{algorithmic}
\end{algorithm}
 Now we compute the runtime of the algorithm. For each subset $F\subset E(T)$, in both parts line~\ref{LINECC} runs in time $O(n)$ by a breadth-first search. In Part~1 there are $O(n^{\lfloor3k-3/2\rfloor}/(k-1)!)$ subsets $F$. Also, the sort in line~\ref{LINESORT} runs in time $O(k\ lnk)$. Hence, the running time of Part~1 is of order $O(n^{\lfloor3k-1/2\rfloor}/(k-1)!)$. In Part~2 there are $O(n^{(2k^2-6k-3)}/(k-2)!)$ subsets $F$. Also, in line~\ref{LINEPI} the number of ways to partition $\{A_i\}_1^t$ into $k$ sets is at most $S(2k^2-6k-2,k)$, where $S(n,k)$ is the Stirling number of second kind. Therefore the running time of Part~2 is of order $O((n+S(2k^2-6k-2,k))\ n^{(2k^2-6k-3)}/(k-2)!)$.
}\end{pproof}
Theorem~\ref{THMFXDk} shows that for every fixed integer $k \geq 2$, computing the $k$th isoperimetric parameters is polynomially solvable for weighted trees. However, the following theorem shows that this result can not be generalized to the case of weighted graphs with bounded tree-width. This can also be considered as a generalization of Papadimiriu's result (Theorem~\ref{THMCMPLX}(ii)).
\begin{thm}\label{THMCGG}
For every fixed integer $k\geq 2$, {\rm IPP}$_k$ and {\rm NCP}$_k$ $($in both max and mean versions$)$ are $NP$-complete for bipartite weighted graphs with tree-width $2$.
\end{thm}
\begin{proof}{
First we show that it is enough to prove the theorem for $k=2$. For this, assume that $k>2$ is an integer and $G$ is a weighted graph. Add $k-2$ new isolated vertices of weight $1$ to obtain a new weighted graph $G'$. For every $k$-subpartition of $V(G')$, there are two subsets completely included in $V(G)$. Thus, solving IPP$_2$ (equivalently NCP$_2$) for the graph $G$ is equivalent to solving IPP$_k$ and NCP$_k$ for the graph $G'$. Henceforth, we concentrate on NCP$^M_2$, mentioning that the proof of the mean version is similar.

Consider the following $NP$-complete problem in the class of KNAPSACK problems, known as the PARTITION problem \cite{GJ79}.
\begin{prob}
{PARTITION}
{$n$ positive integers $x_1,\ldots,x_{n}$ such that
$\sum_{i=1}^{n} x_i=2B$.}
{Is there a subset $I\subset [n]$ such that
$\sum_{i\in I} x_i=B$?}
\end{prob}
We shall propose a polynomial reduction from PARTITION to NCP$^M_2$.
Let $x_1,\ldots,x_{n}$ be $n$ positive integers where
$\sum_{i=1}^{n} x_i=2B$. Then, define the bipartite weighted graph $G$ as follows.
\[V(G):=\{u_1,u_2,v_1,\ldots,v_{n}\},\quad E(G):=\{u_1v_i, u_2v_i, 1\leq i\leq n\},\] \[\omega(u_1)=\omega(u_2)=M,\quad \omega(v_i):=x_i, \ \forall\ 1\leq i\leq n,\]
where $M$ is an arbitrary positive integer. Also, let all the edge weights be equal to 1. It is clear that the graph $G$ has tree-width equal to $2$.
 Assume $M$ is sufficiently larger than $B$, then by Lemma~\ref{LEMDIF}, in every minimizing 2-partition achieving $\tilde{\iota}_2(G)$, $u_1$ and $u_2$ are in different parts. Thus, if $(A_1,A_2)$ is a minimizing 2-partition, then
\[\tilde{\iota}^M_2(G)=\max\left\{\frac{c(A_1)}{\omega(A_1)},\frac{c(A_2)}{\omega(A_2)}\right\}=\max\left\{\frac{n}{M+\sum_{v_i\in A_1} x_i}, \frac{n}{M+\sum_{v_i\in A_2} x_i}\right\}.\]
Hence, $\iota^M_2(G)\leq n/(M+B)$ if and only if $\sum_{v_i\in A_1} x_i=B$. This completes the proof.
}\end{proof}
\section{The Unitarization Process\label{SECUNIT}}
In this section  we establish a machinery to convert the hardness results from weighted graphs to unweighted (simple) graphs, i.e. graphs whose all the vertex and edge weights are equal to $1$. In fact this method that we call the {\it unitarization} process, is a polynomial reduction and will be used to prove some hardness results for unweighted graphs and trees. Define the class {$\mathcal{ISO}$} to be the set of all problems IPP, NCP, IPP$_k$ and NCP$_k$ for the maximum and mean version.
\begin{prop}\label{PROUNIT}
If $P$ is a problem in the class $\mathcal{ISO}$ which is $NP$-complete in the strong sense for weighted graphs, then it is also $NP$-complete for  unweighted $($simple$)$ graphs.
\end{prop}
\begin{proof}
{We prove the proposition for NCP$^M$ and the other cases are similar. Assume that NCP$^M$ is $NP$-complete in the strong sense and let $G=(V,E,\omega,c)$ together with the integer $k\geq 2$ and the number $N=M/L$ be an instance of NCP$^M$, where all the weights and integers $M,L$ are given in unary codes. We apply a {\it unitarization process} on $G$ which is a polynomial reduction to obtain a simple graph $G'$  with all the weights equal to $1$  and a constant $N'$, such that for NCP$^M$, $(G,k,N)$ is a positive instance if and only if $(G',k,N')$ is a positive instance. This implies the $NP$-completeness of NCP$^M$ for unweighted graphs. The process is described in two steps.\\

\noindent{\bf Step 1.} Unitarization of the vertex weights.\\
In this step, we propose a method to make all the vertex weights equal to 1. First, multiply all the vertex weights by a sufficiently large constant $\chi$ such that for every vertex $u\in V$, $\chi \omega(u)\geq \sum_{e=uv\in E} c(e)$. Then, for every $A\subset V$, we have $c(A)/ \omega(A)\leq \chi$. Now, to construct the graph $G'=(V',E',c)$ from $G$, for each vertex $u\in V$, add a set $W_u$ of exactly $\chi \omega(u)-1$ new vertices and join all of the vertices in $W_u$ to $u$ (see Figure~\ref{FIGB}). Also, let the new edges $ux$, $x\in W_u$, have weights equal to 1. We claim that $\tilde{\iota}^M_k(G)=\chi\tilde{\iota}^M_k(G')$. Let $\{A_i\}_1^k$ be a
minimizing partition for $\tilde{\iota}^M_k(G)$. Then, by defining $A'_i:=A_i\cup (\cup_{u\in A_i} W_u)$, it is clear that ${c(A'_i)}/{|A'_i|}=(1/\chi){c(A_i)}/{\omega(A_i)}$. Therefore, $\tilde{\iota}^M_k(G')\leq(1/\chi) \tilde{\iota}^M_k(G)$. To prove the equality, let $\{B'_i\}_1^k$ be a minimizing  partition achieving $\tilde{\iota}^M_k(G')$. For a vertex $u\in V$, assume $u\in B'_i$, for some $i$. If there exists $x\in W_u$, such that $x\in B'_j$, for some $j\neq i$, then we transfer $x$ from $B'_j$ to $B'_i$ and define $B''_i:=B'_i\cup\{x\}$ and $B''_j:=B'_j\backslash \{x\}$. Since $c(B'_j)/|B'_j|\leq \tilde{\iota}^M_k(G')\leq (1/\chi) \tilde{\iota}^M_k(G)\leq 1$, we have
\[
\frac{c(B''_j)}{|B''_j|}=\frac{c(B'_j)-1}{|B'_j|-1}\leq\frac{c(B'_j)}{|B'_j|},\qquad
\frac{c(B''_i)}{|B''_i|}=\frac{c(B'_i)-1}{|B'_i|+1}\leq\frac{c(B'_i)}{|B'_i|}.
\]
By continuing this process, we get a minimizing partition $\{B''_i\}_1^k$ achieving $\tilde{\iota}^M_k(G')$ with the property that for every vertex $u\in V$, $u$ and the vertices in $W_u$ all are in the same part. Thus, by defining $B_i:=V\cap B''_i$, we have $c(B_i)/\omega(B_i)=\chi c(B''_i)/|B''_i|$. Hence, $\tilde{\iota}^M_k(G)\leq \chi \tilde{\iota}^M_k(G')$. It remains to let $N':=N/\chi$.\\

\noindent{\bf Step 2.} Unitarization of the edge weights.\\
Let $n:=|V|$ and assume that all the vertex weights are equal to $1$ and replace every edge $e\in E$ by exactly $c(e)$ multiple edges. Then subdivide all the edges to obtain a simple graph $G'$ and let the new edge weight function $c'$ be the constant function $1$ (see Figure~\ref{FIGB}). For each edge $e\in E$, let the set of new vertices obtained from the subdivisions be denoted by $S_e$ and define $S:=\cup_{e\in E}S_e$. Also, for a constant $\psi$, define the vertex weight function $\omega'$ to be equal to $1$ on the set $S$ and equal to $\psi$ on the set $V$.
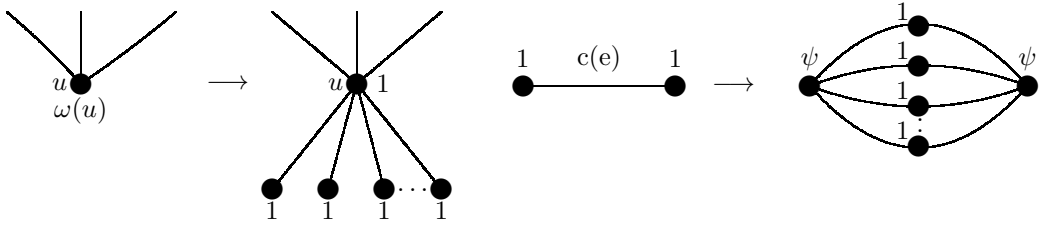
\begin{figure}[ht]
{\def\emline#1#2#3#4#5#6{%
\put(#1,#2){\special{em:moveto}}%
\put(#4,#5){\special{em:lineto}}}
\def\newpic#1{}
\unitlength .4pt
\special{em:linewidth 0.4pt}
\linethickness{0.4pt}
\begin{picture}(300,180)(40,-270)
\bezier{500}(382, -155)(395, -205)(408, -255)
\put(395,-189){\makebox(0, 0)[cc]{}}
\bezier{500}(382, -155)(422, -205)(462, -255)
\put(422,-189){\makebox(0, 0)[cc]{}}
\bezier{500}(382, -155)(368, -205)(355, -255)
\put(368,-189){\makebox(0, 0)[cc]{}}
\bezier{500}(382, -155)(342, -205)(302, -255)
\put(342,-189){\makebox(0, 0)[cc]{}}
\bezier{500}(382, -155)(422, -121)(462, -87)
\put(422,-105){\makebox(0, 0)[cc]{}}
\bezier{500}(382, -155)(342, -121)(302, -87)
\put(342,-105){\makebox(0, 0)[cc]{}}
\bezier{500}(382, -155)(382, -121)(382, -87)
\put(382,-105){\makebox(0, 0)[cc]{}}
\put(382,-155){\circle*{20}}
\put(362,-155){\makebox(0, 0)[cc]{$u$}}
\put(407,-155){\makebox(0, 0)[cc]{1}}
\put(462,-255){\circle*{20}}
\put(462,-275){\makebox(0, 0)[cc]{1}}
\put(408,-255){\circle*{20}}
\put(408,-275){\makebox(0, 0)[cc]{1}}
\put(355,-255){\circle*{20}}
\put(355,-275){\makebox(0, 0)[cc]{1}}
\put(302,-255){\circle*{20}}
\put(302,-275){\makebox(0, 0)[cc]{1}}
\put(435,-255){\makebox(0, 0)[cc]{$\ldots$}}
%
\put(121,-155){\circle*{20}}
\put(101,-155){\makebox(0, 0)[cc]{$u$}}
\put(121,-180){\makebox(0, 0)[cc]{$\omega(u)$}}
\put(260,-155){\makebox(0, 0)[cc]{$\longrightarrow$}}
\bezier{500}(121, -155)(91, -121)(51, -87)
\put(41,-105){\makebox(0, 0)[cc]{}}
\bezier{500}(121, -155)(121, -121)(121, -87)
\put(81,-105){\makebox(0, 0)[cc]{}}
\bezier{500}(121, -155)(171, -121)(211, -87)
\put(101,-105){\makebox(0, 0)[cc]{}}
\end{picture}
\begin{picture}(300,185)(-170,-405)
\bezier{500}(165, -292)(78, -292)(21, -292)
\put(93,-266){\makebox(0, 0)[cc]{c(e)}}
\put(21,-292){\circle*{20}}
\put(21,-266){\makebox(0, 0)[cc]{1}}
\put(220,-292){\makebox(0, 0)[cc]{$\longrightarrow$}}
\put(165,-292){\circle*{20}}
\put(165,-266){\makebox(0, 0)[cc]{1}}
%
\put(292,-292){\circle*{20}}
\put(292,-266){\makebox(0, 0)[cc]{$\psi$}}
\put(498,-292){\circle*{20}}
\put(498,-266){\makebox(0, 0)[cc]{$\psi$}}
\bezier{500}(292, -292)(395, -175)(498, -292)
\bezier{500}(292, -292)(395, -253)(498, -292)
\bezier{500}(292, -292)(395, -331)(498, -292)
\bezier{500}(292, -292)(395, -409)(498, -292)
\put(395,-235){\circle*{20}}
\put(380,-221){\makebox(0, 0)[cc]{1}}
\put(395,-273){\circle*{20}}
\put(380,-258){\makebox(0, 0)[cc]{1}}
\put(395,-311){\circle*{20}}
\put(380,-296){\makebox(0, 0)[cc]{1}}
\put(395,-349){\circle*{20}}
\put(380,-334){\makebox(0, 0)[cc]{1}}
\put(393,-335){$\vdots$}
\end{picture}}
\caption{The vertex and edge gadgets used in the unitarization process.\label{FIGB}}
\end{figure}

We claim that if $\psi$ is sufficiently larger than $n,L$ and $|S|$, then $\tilde{\iota}^M_k(G)\leq N$ if and only if $\tilde{\iota}^M_k(G')\leq (N/\psi)$. For this, first assume that $\tilde{\iota}^M_k(G)\leq N$  and let $\{A_i\}_1^k$ be a minimizing $k$-partition for $\tilde{\iota}^M_k(G)$. Define $A'_i:=A_i\cup (\displaystyle\cup_{e\in E(A_i,A_j), 1\leq j\leq i} S_e)$. It is clear that ${c'(A'_i)}={c(A_i)}$ and ${\omega'(A'_i)}\geq\psi{|A_i|}$. Therefore, $\tilde{\iota}^M_k(G')\leq(N/\psi)$. On the other hand, assume that $\tilde{\iota}^M_k(G')\leq(N/\psi)$ and let $\{B'_i\}_1^k$ be a minimizing  partition achieving $\tilde{\iota}^M_k(G')$. By defining $B_i:= B'_i\cap V$, we have $c(B_i)\leq c'(B'_i)$. Also, if $\psi$ is sufficiently larger than $n,L$ and $|S|$, then
\[\frac{c(B_i)}{|B_i|}<\frac{\psi c(B_i)}{\psi |B_i|+|S|}+\frac{1}{nL}\leq\frac{\psi c'(B'_i)}{\omega'(B'_i)}+\frac{1}{nL}.\]
Thus,
\[\max_{1\leq i\leq k}\left\{\frac{c(B_i)}{|B_i|}\right\}=\frac{c(B_{i_0})}{|B_{i_0}|}<
\psi\ \tilde{\iota}^M_k(G')+\frac{1}{nL}\leq \frac{M}{L}+\frac{1}{nL} \leq \frac{M}{L}+\frac{1}{L|B_{i_0}|}.\]
And consequently,
$\tilde{\iota}^M_k(G)\leq\max_{i}({c(B_i)}/{|B_i|})\leq {M}/{L}=N.$
This completes the second step.

Finally, by repeating Step 1, we may find a simple graph all of whose edge and vertex weights are equal to 1.
Note that since the edge and vertex weights of $G$ and also $M,L$ are given in unary codes, the obtained simple graph is polynomial time computable.
}\end{proof}
By Theorem \ref{THMCMPLX}~(i), we know that NCP$_2$ is $NP$-complete for graphs with multiple edges. Thus, NCP$_2$ is $NP$-complete in the strong sense for weighted graphs. The following corollary is deduced from this fact along with Theorem~\ref{THMNCPM} and Proposition~\ref{PROUNIT}. Part~(i) can be seen as a generalization of Mohar's result (Theorem~\ref{THMCMPLX}~(i)).
\begin{cor}\\
{\rm (i)} For every fixed $k\geq2$, {\rm IPP}$_k$ and {\rm NCP}$_k$ $($in both max and mean versions$)$ are $NP$-complete for unweighted $($simple$)$ graphs.\\
{\rm (ii)} The problem {\rm NCP}$^M$ is $NP$-complete for unweighted trees.
\end{cor}

\section{Concluding Remarks}
Our results show that the study of isoperimetric numbers and minimum normalized cuts on weighted trees is not only  important because of its wide range of applications, but also the scope of weighted trees provide a very interesting arena to test the computational complexity of these problems in which
these isoperimetric problems change their computational behavior by a very slight perturbation of conditions. This fact, on the one hand, is quite interesting from a complexity theoretic point of view, where one is quite interested to investigate problems close to the borders of the classes $P$  and $NP$-complete, and on the other hand, is also interesting from the point of view of approximation algorithms for applications.
In this regard, according to our results, {\it intuitively}, passing from taking maximums to the mean or restricting
the space of subpartitions to partitions will generally make the problem computationally harder. These observations provide enough evidence for the fact that the study of the following open problems ought to be interesting.
\begin{itemize}
\item{Does there exist a polynomial time algorithm that given the number $k \geq 2$ and a weighted tree $T$, computes the parameter $\iota^M_k(T)$? }
\item{Given a constant number $k \geq 2$, does there exist a polynomial time algorithm that computes the parameter $\tilde{\iota}^m_k$ for weigthed trees?}
\item{Determine the computational complexity of IPP$^m$ and NCP$^m$ for unweighted trees.}
\item{Determine the computational complexity of IPP$_k$ and NCP$_k$ for bipartite planar unweighted graphs.}
\end{itemize}
Also, one may consider a number of different variants of isoperimetric problems on graphs and study their computational properties.
As a couple of these variants one may propose the following setups.

Firstly, considering  the maximum and mean versions of the introduced parameters as $\|.\|_{_\infty}$ and $\|.\|_{_1}$ counterparts of the isoperimetric problem, respectively, it is interesting to study the $\|.\|_{_p}$ versions of these parameters and the computational complexity of the corresponding problems. In this setting, it is important to try to characterize the properties that are responsible for the change of hardness
from $NP$-completeness of IPP$^m$ to the polynomial solvability of IPP$^M$ in the limit.

Secondly, one may consider the supervised version of the partitioning problems and formulate them as  multiterminal isoperimetric problems,
in which given a weighted graph along with $k$ specified vertices $v_{1}, \ldots, v_{k}$, we look for a $k$-subpartition ($k$-partition)
such that $v_i$'s appear in different parts and the corresponding cost functions (see Definition~\ref{DEFISO}) are minimized.
For instance, using a similar argument as in the proof of Theorem~\ref{THMCGG}, one may prove that for any $k \geq 2$ the multiterminal versions of IPP$_k$ and NCP$_k$ are $NP$-complete for weighted trees\cite{JAV10}.

As another variant of these problems, one may focus on the approach through $(k,b)$-subpartitions (see \cite{NAIB08,NK07}) that can be considered as
a combination of the max and the mean approach and follow the same line of study.



\end{document}